\documentclass[conference]{IEEEtran}
\IEEEoverridecommandlockouts

\usepackage{cite}
\usepackage{amsmath,amssymb,amsfonts}
\usepackage{algorithmic}
\usepackage{graphicx}
\usepackage{textcomp}
\usepackage{xcolor}
\def\BibTeX{{\rm B\kern-.05em{\sc i\kern-.025em b}\kern-.08em
    T\kern-.1667em\lower.7ex\hbox{E}\kern-.125emX}}

\usepackage{courier} 
\usepackage{braket} 
\usepackage{mathtools} 
\usepackage{arydshln}    
\usepackage{fancyhdr}
\newtheorem{theorem}{Theorem}
\newtheorem{proof}{Proof}

\newtheorem{definition}{Definition}

\begin{document}

\title{Entanglement in Higher-Radix Quantum Systems}

\author{\IEEEauthorblockN{Kaitlin N. Smith and Mitchell A. Thornton}
\IEEEauthorblockA{Quantum Informatics Research Group\\
Southern Methodist University, 
Dallas, Texas, USA\\
\{knsmith, mitch\}@smu.edu}
}

\maketitle
\thispagestyle{fancy}
\begin{abstract}
Entanglement is an important phenomenon that enables quantum information processing algorithms and quantum communications protocols. Although entangled quantum states are often described in  radix-2, higher-radix qudits can become entangled as well. In this work, we both introduce partial entanglement, a concept that does not exist for radix-2 quantum systems, and differentiate between partial and maximal entanglement within non-binary quantum information processing systems. We also develop and present higher-radix maximal entanglement generator circuits that are analogous to the well-known Bell state generator for binary quantum systems.   Because higher-dimensioned qudits can be subjected to entangling processes that result in either partially or maximally entangled states, we demonstrate how higher-radix qudit circuits can be composed to generate these varying degrees of partial quantum entanglement.  Theoretical results are provided for the general case where the radix is greater than two, and specific results based on a pair of radix-4 qudits are described.
\end{abstract}

\begin{IEEEkeywords}
qudit, entanglement, partial entanglement, maximal entanglement, quantum information science, QIS, quantum information processing, QIP
\end{IEEEkeywords}

\section{Introduction}

Entanglement is an important aspect in quantum information science (QIS) that is unseen in conventional computing.  For example, if a programmer writes a bit sequence to a register \texttt{reg0}, it is certainly not expected that this write operation would cause the content in another register, such as \texttt{reg1}, to also change.  However, this is precisely the action that would occur in a quantum information processing (QIP) system if \texttt{reg0} and \texttt{reg1} collectively contained a set of qudits in a mutually entangled state.  While the benefits of entanglement may not be readily apparent, it is exploited in numerous algorithms such as Shor's factoring \cite{Sh:94}, Grover's search \cite{Gr:96}, and others in QIS communications systems and protocols like BB84 and its derivatives.  Additionally, entanglement over a long distance was recently demonstrated when the Micius satellite produced an entangled pair of photons in space and transmitted one of the pair to Earth while the entangled state with the remaining photon on the satellite was maintained \cite{YC+:17}.   

The aforementioned QIS algorithms and the Chinese experiment in \cite{YC+:17} were based on the use of quantum bits or ``qubits'' that have a binary basis, typically the so-called ``computational basis,'' of $\{\Ket{0},\Ket{1}\}$.  There has been considerably less research and development on QIS systems based upon higher-dimensional quantum digits, or ``qudits.''  Qudits are simply a generalization of qubits that can be mathematically expressed as a linear combination of an extended computational basis set.  For example, an $R$-dimensional qudit would have a corresponding computational basis set of $\{\Ket{0},\Ket{1}, \cdots, \Ket{r^n-1}\}$ where $R=r^n$, $r>2$, and $(r,n)\in \mathbb{Z}$. 

Entangled states have been studied in the past for radix-3, -4, and -5 QIP systems \cite{enriq_2016}, but the required operators for generating higher-radix entangled states have not been clearly outlined. Additionally, higher-radix QIP demonstrates the unique ability to generate partially entangled states, and achieving these variations of partial entanglement through a set of quantum operators can also be generalized. In this paper, methods for generating higher-radix partial and maximal entanglement will be discussed. A radix-4 system will be used to demonstrate the described techniques.
 
\section{Background QIP Concepts}

\subsection{Information and Operations}
Due to the postulates and axioms of quantum electrodynamic theory, the
quantum state of a wavefunction is modeled as a vector over a discrete Hilbert vector space.  The most
common unit of quantum information is thus represented as a two-dimensional vector in $\mathbb{H}_2$
and is referred to as the qubit.  Qubits are expressed as a linear combination of two basis vectors that span $\mathbb{H}$.  The most common
basis is the so-called computational basis represented as $\Ket{0_2}$ and $\Ket{1_2}$.  We note that
subscripts are used to refer to the value $r$ to avoid confusion while discussing systems of
different values of $r$.  A qubit is expressed as a linear combination of the two orthonormal
basis states,  $\Ket{0_2}=\left[ \begin{matrix}
  1 & 0  
\end{matrix} \right]^{\rm T}$ and $\Ket{1_2}=\left[ \begin{matrix}
  0 & 1  
\end{matrix} \right]^{\rm T}$. In general, a qubit is represented as

\begin{equation}\label{eq:qubit_state}
\Ket{\phi_2} = \alpha\Ket{0_2}+\beta\Ket{1_2}
\end{equation}

\noindent where $\alpha$ and $\beta$ are complex-valued probability amplitudes. For the qubit $\Ket{\phi_2}$, the probability that $\Ket{\phi_2}=\Ket{0_2}$ after measurement is $\alpha^*\alpha= |\alpha|^2$ and the probability that $\Ket{\phi_2}=\Ket{1_2}$ after measurement is $\beta^*\beta= |\beta|^2$ where $|\alpha|^2+|\beta|^2=1$.  

For higher-dimensional systems, the quantum digit is modeled
as a vector in $\mathbb{H}_r$ where $r>2$. The dimension of the overall QIP system is $R=r^n$ where $n$ indicates the amount of included qudits.   A radix-$r$ qudit is a linear combination of $r$ basis states in the form of 

\begin{equation}\label{eq:gen_qudit}
\Ket{\phi_r} = \sum_{i=0}^{r-1} a_i \Ket{i_r}.
\end{equation}

\noindent Using Eqn.~\ref{eq:gen_qudit} to generate a radix-4 qudit results in

\begin{equation}\label{eq:qudit_state}
\Ket{\phi_4} = a_0\Ket{0_4}+ a_1\Ket{1_4} + a_2\Ket{2_4}+ a_3\Ket{3_4} .
\end{equation}

\noindent In Eqn.~\ref{eq:qudit_state}, the basis states used are $\Ket{0_4}=\left[ \begin{matrix}
  1 & 0 & 0 & 0 
\end{matrix} \right]^{\rm T}$, $\Ket{1_4}=\left[ \begin{matrix}
  0 & 1 & 0 & 0 
\end{matrix} \right]^{\rm T}$, $\Ket{2_4}=\left[ \begin{matrix}
  0 & 0 & 1 & 0 
\end{matrix} \right]^{\rm T}$, and $\Ket{3_4}=\left[ \begin{matrix}
  0 & 0 & 0 & 1 
\end{matrix} \right]^{\rm T}$, and $a_0$, $a_1$, $a_2$, and $a_3$ are probability amplitudes. The technological reasons prohibiting the widespread adoption of higher-radix classical computation are not present in the realization of QIP systems. Electronic information processing systems utilize the binary digit, or ``bit,'' as a basic unit of information, and a bit is typically represented by a voltage range.  Due to the rapid advances in transistor scaling versus allowable voltage levels, the bit has dominated as the preferred choice in classical computing.  Within QIS, there are no clear preferences with regard to the choice of information carriers.  
While many current QIP implementations are binary, or based upon qubits, higher-dimensional systems have certainly not been ruled out. Some examples of qudit-based photonic QIP systems that are found in the literature encode information with orbital angular momentum (OAM), time-energy, frequency, time-phase, and location.

For any radix-$r$ quantum system, computation is accomplished with quantum operations. These operators, also referred to as gates, are each represented by a unitary transfer function matrix, $\mathbf{U}$. The matrix $\mathbf{U}$ is of size $r^n\times r^n$ where $n$ indicates the number of radix-$r$ qudits being transformed.

\subsection{Superposition}
A qubit or qudit can comprise non-zero probability amplitudes for all of its basis states simultaneously. This characteristic is known as quantum ``superposition.'' Superposition allows quantum computing algorithms to parallelize computations since a state holds multiple values at once. In quantum systems, parallelized computation is achieved through
parallelism in information representation rather than the spatial or temporal (\textit{i.e.} pipelining) parallelism employed in classical systems.

\begin{definition} \textit{Maximal Superposition} \\ \label{def:max-super}
\noindent A state is is said to be ``maximally superimposed''  when the state vector can be expressed as a linear combination of all basis vectors such that each probability amplitude has an identical magnitude that is equal to $\frac{1}{r}$. In other words, the quantum state takes the value of $ \Ket{\phi_r} = \sum_{i=0}^{r-1} a_i  \Ket{i_r}$ where each $|a_i|=a_i^*a_i=a_i a_i^*=\frac{1}{r}$.\hfill $\Box$

\end{definition}

\subsection{Entanglement}
When two or more quantum entities are entangled, they interact and behave as a single composite system. Operations and measurements performed on one portion of the entangled group directly impacts the state of the other entangled elements. Entanglement is a significant phenomenon as it powers quantum communication techniques. For example, entanglement enables quantum teleportation, a highly secure communication protocol where quantum states can be obtained in a near simultaneous manner as demonstrated in~\cite{YC+:17}.

\section{Entanglement Generation}
Two quantum units of any radix can only become entangled if the proper operators are applied. First, the pair is initialized into a basis state.  Next, one of the units is transformed into a state of maximal superposition. Finally, the pair must interact with each other using one or more two-input operators. 

\subsection{The Chrestenson Gate}

The Hadamard gate is represented with the transfer matrix

\begin{equation}\label{eq:Hadamard}
{\scriptsize
\mathbf{H}=\frac{1}{\sqrt{2}}
\arraycolsep=1.4pt\def\arraystretch{1.4}
\begin{bmatrix*}[r]
    1 & 1 \\[-4pt]
    1 & -1
\end{bmatrix*}}.
\end{equation}

\noindent This gate is an important operator used to evolve a qubit into a maximally superimposed state. When a qubit is initially in a basis state, $\mathbf{H}$ transforms the qubit so that it has an equal probability of being measured as either $\Ket{0_2}$ or $\Ket{1_2}$. Quantum operators exist for many different computation bases, such as radix-3 and above, that achieve equal, and therefore maximal, superposition among the corresponding basis states. These operators are derived using the discrete Fourier transform on Abelian groups. General theory of Fourier transforms on Abelian groups is outlined in the literature \cite{vilenkin, b6}. The multiple-valued generalization of the radix-2 Hadamard gate and its transfer matrix is composed of discretized versions of the orthogonal Chrestenson basis function set \cite{b6}. This QIP gate is generally referred to as the Chrestenson gate \cite{b7}. Examples of useful applications of the Chrestenson transform in QIP can be found in reference \cite{Chrest_gate_examples}.

The Chrestenson operator, as the generalized version of the Hadamard operator, can be defined by a matrix that is parameterized depending upon the on the radix of computation. The resulting radix-$r$ Chrestenson transformation matrix for a single qudit has a size of $r \times r$, with column (or row) vectors that are orthonormal to one another since the matrix is unitary. Traditionally, the Chrestenson transformation matrix is expressed as being normalized with a scalar factor, $\frac{1}{\sqrt{r}}$, thus permitting each of the components within the Chrestenson transform matrix to take the form of one of the $r^{th}$ roots of unity raised to some integral power \cite{b6,b7}. The $r^{th}$ roots of unity can be visualized as $r$ points that are evenly-spaced on the unit circle in the complex plane with one of the roots always being the real-valued unity value or +1. The roots of unity are indicated as $w_k$ where $k=0, 1, ..., (r-1)$, and the point (1,0), denoted as $w_0$, is always included in this set.  Each root of unity satisfies $(w_k)^r=1$. A closed-form representation of the $r^{th}$ roots of unity is $w_k= e^{i\frac{2\pi}{r}\times k}$.

The structure of the Chrestenson transform matrix takes the form of a Vandermonde matrix where each row vector consists of a $r^{th}$ root of unity, $w_k$, raised to an integral power $j$. Each element of the matrix is $w_k^j$ where $j$ is the column index and $k$ is the row index. In this indexing system, $j=0$ represents the leftmost column vector and $j=(r-1)$ represents the rightmost column vector. Similarly, $k=0$ represents the topmost row vector and $k=(r-1)$ represents the bottommost row vector. The Hadamard matrix results from the Chrestenson transform matrix when $r = 2$, confirming that the Chrestenson transform is a generalization of the Hadamard transform. The generalized radix-$r$ Chrestenson transform matrix, $\mathbf{C}_r$, is represented with the matrix

\begin{equation}\label{eq:gen_Chrestenson}
{\scriptsize
\mathbf{C}_r =
\frac{1}{\sqrt{r}}
\arraycolsep=1.4pt\def\arraystretch{1.6}
\begin{bmatrix*}[c]
    w^{0}_{0} & w^{1}_{0} & \dots & w^{(r-1)}_{0} \\[-2pt]
    w^{0}_{1} & w^{1}_{1} &  \dots & w^{(r-1)}_{1} \\[-2pt]
    \vdots & \vdots & \ddots & \vdots  \\[-2pt]
    w^{0}_{(r-1)} & w^{1}_{(r-1)}  & \dots & w^{(r-1)}_{(r-1)}   
\end{bmatrix*}}.
\end{equation}

\noindent Using the fourth roots of unity, $w_0 = \exp[(i2\pi/4)*0] = 1$, $w_1 =\exp[(i2\pi/4)*1] = i$, $w_2 = \exp[(i2\pi/4)*2] = -1,$ and $w_3 = \exp[(i2\pi/4)*3] = -i$, in Eq.~\ref{eq:gen_Chrestenson}, the radix-4 Chrestenson gate transfer matrix becomes 

\begin{equation}\label{eq:rad-4_Chrestenson}
{\scriptsize
\mathbf{C_4}=
\frac{1}{\sqrt{4}}
\arraycolsep=1.4pt\def\arraystretch{1.2}
\begin{bmatrix*}[r]
    1 & 1 & 1 & 1 \\[-2pt]
    1 & i & -1 & -i \\[-2pt]
    1 & -1 & 1 & -1\\[-2pt]
    1 & -i & -1 & i 
\end{bmatrix*}}. 
\end{equation}

The radix-4 Chrestenson gate ($\mathbf{C}_4$), allows a radix-4 qudit originally in a basis to evolve into a quantum state of equal superposition. More information about the radix-4 Chrestenson gate as well as a proposed physical implementation can be found in \cite{rad4chrest}. The following example shows how the radix-4 qudit $\Ket{a_4} = \Ket{0_4}$ evolves to $\Ket{b_4} = \frac{1}{2}\Ket{0_4} + \frac{1}{2}\Ket{1_4} + \frac{1}{2}\Ket{2_4} + \frac{1}{2}\Ket{3_4}$, taking the value of the first column of the radix-4 Chrestenson matrix, after it is applied to the $\mathbf{C_4}$ transform

\[
\mathbf{C}_4\Ket{a_4} = \ket{b_4},
\]

\[ 
{\scriptsize
\mathbf{C}_4\Ket{0_4} =
\frac{1}{\sqrt{4}}  
\arraycolsep=1.4pt\def\arraystretch{1.2}
\begin{bmatrix*}[r]
    1 & 1 & 1 & 1\\
    1 & i & -1 & -i \\
    1 & -1 & 1 & -1\\
    1 & -i & -1 & i \\
\end{bmatrix*}\begin{bmatrix}
    1 \\
    0 \\
    0 \\
    0 \\
\end{bmatrix} =\frac{1}{2}\begin{bmatrix}
    1 \\
    1 \\
    1 \\
    1 \\
\end{bmatrix}},
\]
\[
{\scriptsize
\mathbf{C}_4\Ket{0_4} = \frac{1}{2}[\Ket{0_4} + \Ket{1_4} + \Ket{2_4} + \Ket{3_4}].
}
\]

\noindent The schematic symbol of the $\mathbf{C_4}$ gate is pictured in Fig.~\ref{fig:C4}.

\begin{figure}[!t]
\includegraphics[height=.45in]{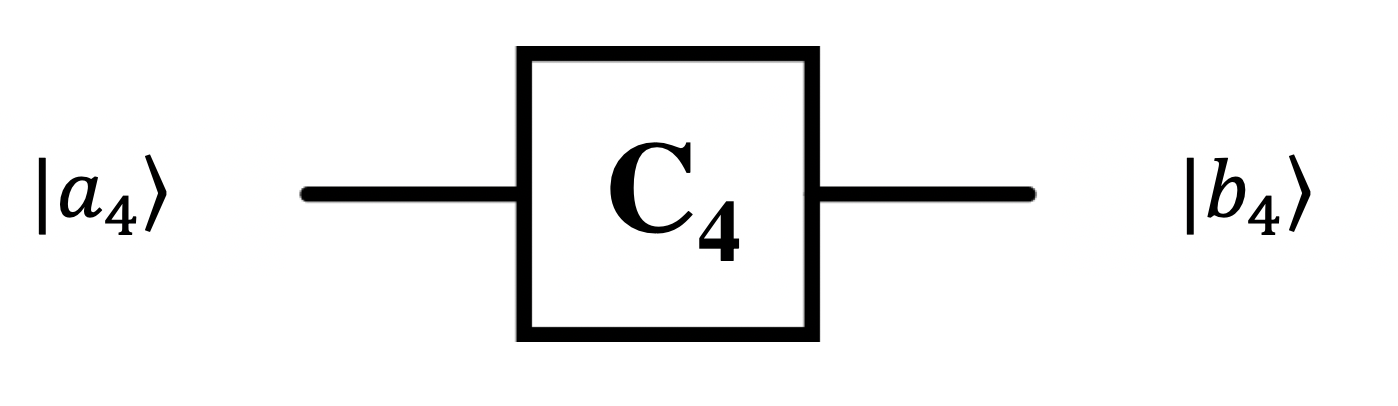}
\centering
\caption{Symbol of the radix-4 Chrestenson gate, $\mathbf{C_4}$.}
\label{fig:C4}
\end{figure}

\subsection{The Controlled Modulo-Add Gate }
The radix-2 $\mathbf{X}$ operator, or the $\mathbf{NOT}$ operator, performs a Pauli-$\mathbf{X}$ rotation on a qubit. Mathematically, the Pauli-$\mathbf{X}$ can be considered a modulo-2 addition-by-one operator
since it evolves a qubit $\Ket{0_2}$ to be $\Ket{((0+1) \text{mod} \; 2)_2}=\Ket{1_2}$ and likewise evolves a qubit
$\Ket{1_2}$ to $\Ket{((1+1) \text{mod} \; 2)_2}=\Ket{0_2}$.  In the case where $\Ket{\phi_2}$ is in a state
of superposition, $\Ket{\phi_2(t_0)}=a_0 \Ket{0_2} + a_1 \Ket{1_2}$, the $\mathbf{X}$ operation exchanges the probability amplitude coefficients of the quantum state yielding the qubit $\Ket{\phi_2(t_1)}=a_1\Ket{0_2}+a_0\Ket{1_2}$.
The quantum gate or operator for the Pauli-$\mathbf{X}$ is represented with the transfer matrix

\begin{equation}\label{eq:NOT}
{\scriptsize
\mathbf{X} = \arraycolsep=1.2pt\def\arraystretch{1.4}
\begin{bmatrix*}[r]
    0 & 1 \\[-5pt]
    1 & 0
\end{bmatrix*}}.
\end{equation}

\noindent The controlled version of the $\mathbf{X}$ gate is the ``controlled-$\mathbf{X}$'' or
``controlled-$\mathbf{NOT}$'' gate denoted as $\mathbf{C}_{NOT}$.  The controlled-$\mathbf{NOT}$ gate may also be referred to by the somewhat unconventional
name of ``controlled-modulo-add by one'' gate. The $\mathbf{C}_{NOT}$ gate is defined as

\begin{equation}\label{eq:CNOT}
{\scriptsize
\mathbf{C}_{NOT} = 
\arraycolsep=1.2pt\def\arraystretch{1.4}
\begin{bmatrix*}[r]
    1& 0 & 0 & 0 \\[-5pt]
    0 & 1 & 0 & 0 \\[-5pt] 
    0& 0 & 0 & 1 \\[-5pt]
    0 & 0 & 1 & 0
\end{bmatrix*}},
\end{equation}

\noindent and it causes a target qubit to undergo a Pauli-$\mathbf{X}$ operation if the control qubit has a probability amplitude for $\Ket{1_2}$. 

In the case of radix-2 qubit systems, only two different modulo-$k$ additions are possible, modulo-0 and modulo-1, since there
are only two computational basis vectors.
Furthermore, one of these
is the trivial case of modulo-2 addition-by-zero that results in the identity transfer matrix. 

We denote the single qudit modulo-addition operations as $\mathbf{M}_k$
for operators that cause a modulo-$k$ addition with respect to modulus-$r$ as is used in \cite{qlog_unit_ops}.  As is the case with qubits (\textit{i.e.}, $r=2$), the modulo-0 operation 
is equal to the identity transfer function, or $\mathbf{M}_0=\mathbf{I}_r$ where $\mathbf{I}_r$ is the $r \times r$ identity matrix. 
As an example, the non-trivial radix-4 $\mathbf{M}_k$ operators are

\mbox{ \small
\begin{math}
\begin{array}{ccc}
\arraycolsep=0.2pt\def\arraystretch{1.0}
\begin{array}{ccc}
\mathbf{M}_{1} & = &  
\left[ 
\arraycolsep=1.4pt\def\arraystretch{1.0}
\begin{array}{cccc}
    0 & 0 & 0 & 1 \\[-2pt]
    1 & 0 & 0 & 0 \\[-2pt]
    0 & 1 & 0 &0 \\[-2pt]
    0 & 0 & 1 & 0 
\end{array} \right] ,
\end{array}
&
\arraycolsep=0.2pt\def\arraystretch{1.0}
\begin{array}{ccc}
\mathbf{M}_{2} & = & 
\left[ 
\arraycolsep=1.4pt\def\arraystretch{1.0}
\begin{array}{cccc}
    0& 0 & 1 & 0 \\[-2pt]
    0& 0 & 0 & 1 \\[-2pt] 
    1& 0 & 0 &0 \\[-2pt]
    0 & 1 & 0 & 0 
\end{array} \right],
\end{array}
&
\arraycolsep=0.2pt\def\arraystretch{1.0}
\begin{array}{ccc}
\mathbf{M}_{3} & = & 
\left[ 
\arraycolsep=1.4pt\def\arraystretch{1.0}
\begin{array}{cccc}
    0& 1 & 0 & 0 \\[-2pt]
    0 & 0 & 1 & 0 \\[-2pt] 
    0& 0 & 0 &1 \\[-2pt]
    1 & 0 & 0 & 0 
\end{array} \right].
\end{array}
\end{array}
\end{math}}

For higher-dimensional systems with radix-$r$, $r>2$, there are $r-1$ different 
non-trivial single-qudit modulo-$k$ additions.   
Considering all available control values for radix-$r$,  as well as the different modulo-$k$ additions, there are a total 
of $r^2-r$ different and non-trivial controlled-modulo-addition operators. 
The radix-$r$ controlled Modulo-addition-$k$ matrix, $\mathbf{A}_{h,k}$ where $h$ and $k$ each 
contain a single value, takes the form of

\begin{equation} \label{eq:gen-cont-mod-add} 
{\scriptsize
\arraycolsep=1.4pt\def\arraystretch{1.0}
\begin{array}{c}
\mathbf{A}_{h,k}=
\left[
\arraycolsep=1.4pt\def\arraystretch{1.0}
\begin{array}{ccccccc}
\mathbf{D}_0 & \mathbf{0}_r & \cdots & \cdots & \cdots & \cdots & \mathbf{0}_r \\
\mathbf{0}_r & \mathbf{D}_1 & \mathbf{0}_r & \cdots & \cdots& \cdots &  \mathbf{0}_r  \\
\vdots  &  \mathbf{0}_r  &  \ddots & \mathbf{0}_r &  \cdots &  \cdots & \mathbf{0}_r \\
\vdots & \vdots &  \mathbf{0}_r & \mathbf{D}_j & \mathbf{0}_r & \cdots & \mathbf{0}_r \\
\vdots & \vdots &  \vdots & \mathbf{0}_r &  \ddots & \mathbf{0}_r & \vdots \\
\vdots & \vdots &  \vdots & \vdots & \mathbf{0}_r &  \ddots & \mathbf{0}_r \\
\mathbf{0}_r & \mathbf{0}_r &  \mathbf{0}_r & \mathbf{0}_r & \cdots & \mathbf{0}_r & \mathbf{D}_{(r-1)}
\end{array}
\right], 
\\
\text{where,   } \mathbf{D}_i = \begin{cases}
\mathbf{M}_0 = \mathbf{I}_r, & i \neq h \\
\mathbf{M}_k, & i=h.
\end{cases}
\end{array}}
\end{equation}
 
\noindent In Eqn.~\ref{eq:gen-cont-mod-add}, each submatrix along the diagonal is denoted as $\mathbf{D}_i$ and is of dimension $r \times r$. The two-qudit controlled variation of the modulo-add gate, $\mathbf{A}_{h,k}$, only allows the modulo-addition by $k$ operation to occur on the target whenever the control qudit is in state, $\Ket{h_r}$. For example,

\begin{equation}\label{eq:r4_A_3,1}
{\scriptsize
\begin{array}{c}
\mathbf{A}_{3,1}=
\left[
\arraycolsep=1.4pt\def\arraystretch{1.0}
\begin{array}{rrrr;{2pt/2pt}rrrr;{2pt/2pt}rrrr;{2pt/2pt}rrrr}
    1& 0 & 0 & 0 & 0 &0 &0 &0 &0 & 0 & 0 &0 &0 &0 &0 &0\\
    0& 1 & 0 & 0 & 0 &0 &0 &0 &0 & 0 & 0 &0 &0 &0 &0 &0\\
    0& 0 & 1 & 0 & 0 &0 &0 &0 &0 & 0 & 0 &0 &0 &0 &0 &0\\
    0& 0 & 0 & 1 & 0 &0 &0 &0 &0 & 0 & 0 &0 &0 &0 &0 &0\\\hdashline[2pt/2pt]
    0& 0 & 0 & 0 & 1 &0 &0 &0 &0 & 0 & 0 &0 &0 &0 &0 &0\\
    0& 0 & 0 & 0 & 0 &1&0 &0 &0 & 0 & 0 &0 &0 &0 &0 &0\\
    0& 0 & 0& 0 & 0 &0 &1 &0 &0 & 0 & 0 &0 &0 &0 &0 &0\\
    0& 0 & 0 & 0& 0 &0 &0 &1 &0 & 0 & 0 &0 &0 &0 &0 &0\\\hdashline[2pt/2pt]
    0& 0 & 0 & 0 & 0 &0 &0 &0 &1 & 0 & 0 &0 &0 &0 &0 &0\\
    0& 0 & 0 & 0 & 0 &0 &0 &0 &0 & 1 & 0 &0 &0 &0 &0 &0\\
    0& 0 & 0& 0 & 0 &0 &0 &0 &0 & 0 & 1 &0 &0 &0 &0 &0\\
    0& 0 & 0 & 0 & 0 &0 &0 &0 &0 & 0 & 0 &1 &0 &0 &0 &0\\\hdashline[2pt/2pt]
    0& 0 & 0 & 0 & 0 &0 &0 &0 &0 & 0 & 0 &0 &0 &0 &0 &1\\
    0& 0 & 0 & 0 & 0 &0 &0 &0 &0 & 0 & 0 &0 &1 &0 &0 &0\\
    0& 0 & 0 & 0 & 0 &0 &0 &0 &0 & 0 & 0 &0 &0 &1 &0 &0\\
    0& 0 & 0 & 0& 0 &0 &0 &0 &0 & 0 & 0 &0 &0 &0 &1 &0\\
\end{array}\right]
\end{array}
=
\left[
\arraycolsep=1.4pt\def\arraystretch{1.0}
\begin{array}{cccc}
\mathbf{D}_0 & \mathbf{0}_4 &  \mathbf{0}_4 & \mathbf{0}_4 \\
\mathbf{0}_4 & \mathbf{D}_1 & \mathbf{0}_4 & \mathbf{0}_4  \\
\mathbf{0}_4  &  \mathbf{0}_4  &  \mathbf{D}_2 & \mathbf{0}_4 \\
\mathbf{0}_4 & \mathbf{0}_4 &  \mathbf{0}_4 & \mathbf{D}_3
\end{array}
\right]
}
\end{equation}

\noindent only allows the $\mathbf{D}_3=\mathbf{M}_1$ operation to execute on the target qudit if the control qudit has a value of $\Ket{3_4}$, or at least a non-zero probability amplitude for $\Ket{3_4}$. The control qudit probability amplitudes for $\Ket{0_4}$, $\Ket{1_4}$, and $\Ket{2_4}$ do not evolve the target because $\mathbf{D}_0=\mathbf{D}_1=\mathbf{D}_2=\mathbf{M}_0 = \mathbf{I}_4$.  In Eqn.~\ref{eq:r4_A_3,1}, the dashed lines separate submatrices so the $\mathbf{D}_i$ values are apparent. The general symbol of the controlled Modulo-add gate is shown in Fig.~\ref{fig:mod-add}.  The total amount of available controlled-modulo add operations, $\mathbf{A}_{h,k}$, varies depending on the radix. There are $r$ possible control values, $h$, in the range $\{0,\cdots,(r-1)\}$, and $r-1$ meaningful values in the range $\{1, \cdots, (r-1)\}$ for the $k$ value in the modulo-addition by $k$ operation. As an example, in a radix-4 QIP system, the controlled-modulo add operations are $\mathbf{A}_{0,1}$, $\mathbf{A}_{0,2}$, $\mathbf{A}_{0,3}$, $\mathbf{A}_{1,1}$, $\mathbf{A}_{1,2}$, $\mathbf{A}_{1,3}$, $\mathbf{A}_{2,1}$, $\mathbf{A}_{2,2}$, $\mathbf{A}_{2,3}$, $\mathbf{A}_{3,1}$, $\mathbf{A}_{3,2}$, and $\mathbf{A}_{3,3}$.  

\begin{figure}[!t]
\includegraphics[height=0.55in]{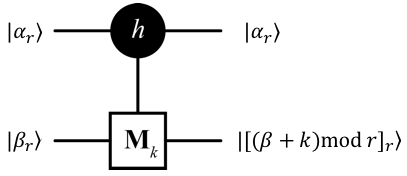}
\centering
\caption{Symbol of the controlled Modulo-add gate, $\mathbf{A}_{h,k}$.}
\label{fig:mod-add}
\end{figure}

\subsection{The Bell State Generator}
The Bell state generator, the inspiration for this work, entangles qubit pairs. An entangled qubit pair, $\Ket{\alpha \beta_2}$ is characterized as a single QIS system wherein the quantum state formed from the qubit pair cannot be factored into a product of $\Ket{\alpha_2} \otimes\Ket{\beta_2}$. 

\begin{definition} \textit{Entangled Qubit Pair} \label{def:qubit-entang} \\
$\Ket{\alpha \beta_2}$ is considered entangled when it consists of two unfactorable basis states with equal and non-zero probability amplitudes. All other basis states in the state vector for the pair have zero-valued probability amplitudes. This leads to four possible entangled two-qubit pairs that are commonly referred to as the ``Bell states.'' The Bell states are 

\begin{equation}\label{eq:Bell}
{\small
\begin{split}
\begin{array}{rrr}
\Ket{B_{00}} = \Ket{\Phi^+} = \frac{\Ket{00_2} + \Ket{11_2}}{\sqrt{2}},&
\mbox{} &
\Ket{B_{01}} = \Ket{\Psi^+} = \frac{\Ket{01_2} + \Ket{10_2}}{\sqrt{2}},
\end{array}\\
\begin{array}{rrr}
\Ket{B_{10}} = \Ket{\Phi^-} = \frac{\Ket{00_2} - \Ket{11_2}}{\sqrt{2}},&
\mbox{} &
\Ket{B_{11}} = \Ket{\Psi^-} = \frac{\Ket{01_2} - \Ket{10_2}}{\sqrt{2}}.
\end{array}
\end{split}
}
\end{equation}
\end{definition}

\noindent The Bell states are created when a pair of qubits, each initialized to a basis state, are operated upon by a quantum algorithm referred to as a Bell state generator. The algorithm may be considered a quantum logic circuit or as a program on a quantum computer, but in either case, it is denoted by the sequence comprised of a Hadamard and $\mathbf{C}_{NOT}$ gate as shown in Fig.~\ref{fig:bell_state_gen}.

\begin{figure}[h!]
\includegraphics [height=0.55in]{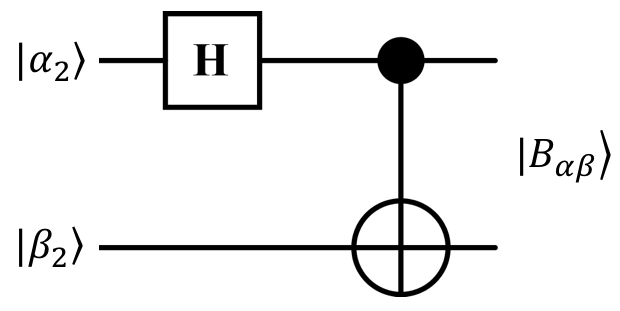}
\centering
\caption{Bell State Generator.}
\label{fig:bell_state_gen}
\end{figure}

As a note, qubit pair entanglement is not limited to just the Bell states. Other arbitrary entangled pairs are possible. These qubit pairs are still entangled, but they contain non-equal magnitudes with respect to the non-zero probability amplitudes in the quantum state vector.  

\section{Demonstration of Partial Entanglement}
Qudits are maximally entangled whenever each of the possible basis state outcomes of an observation are equally likely. Additionally, measurement of one part of the entangled group causes others within the entangled group to collapse without direct observation. In contrast to maximal entanglement, a partially entangled state contains a subset of entangled basis states along with a remaining set that is not entangled. States of partial entanglement should not be confused with non-maximal entanglement where probability amplitudes are imbalanced.  Partial entanglement does not exist in binary QIP systems. The concept of partial entanglement is unique for QIP where $r > 2$. 

\begin{definition} \textit{Partial Entanglement} \label{def:part-ent} \\
Consider a general radix = $r$ qudit expressed in terms of the the computational bases as
$\Ket{\phi_r}=a_0\Ket{0}+a_1\Ket{1}+ \cdots + a_i\Ket{i}+ \cdots + a_{r-1}\Ket{r-1}$.  When
$\Ket{\phi_r}$ is entangled with another radix-$r$ qudit, $\Ket{\theta_r}$, the entangled quantum state
can be expressed as $\Ket{\phi \theta_r}$. If the probability amplitudes of the entangled quantum state
are such that a measurement of $\Ket{\phi_r}$ in a given state, $\Ket{i}$,  implies that $\Ket{\theta_r}$ is
correspondingly in some other given state, $\Ket{j}$ then it is verified that $\Ket{\phi \theta_r}$ are entangled.
However, if the measurement of $\Ket{\phi_r}$ results in an observation of another distinct state, $\Ket{k}$
and furthermore, this measurement does not imply that $\Ket{\theta_r}$ is correspondingly in some
distinct state $\Ket{h}$, then the pair $\Ket{\phi \theta_r}$ are said to be partially entangled. \hfill $\Box$
\end{definition}

Using the structure of the Bell state generator as inspiration, a higher-radix partial entanglement generator can be created with a Chrestenson operator followed by a controlled-modulo-addition operation. An example circuit for radix-4 is seen in Fig.~\ref{fig:part-ent-gen}a. Here, the quantum state $\Ket{00_4}$ is transformed via the partial entanglement generator and evolves to

\begin{equation} \label{eq:T_5a}
\scriptsize{
\begin{split}
\mathbf{T}_{Fig.~\ref{fig:part-ent-gen}a}\Ket{\phi \theta_4}&=\mathbf{T}_{Fig.~\ref{fig:part-ent-gen}a}\Ket{00_4} \\
&=\mathbf{A}_{3,1}\times (\mathbf{C}_4 \otimes \mathbf{I}_4)\Ket{00_4} \\
&=\frac{1}{2}\left[\Ket{00_4} + \Ket{10_4} + \Ket{20_4} + \Ket{31_4} \right] \\
&= \frac{1}{2}\left[ \left(\Ket{0_4} + \Ket{1_4} + \Ket{2_4}\right) \otimes\Ket{0_4}\right] +  \frac{1}{2}\left[ \Ket{31_4} \right]. 
\end{split}}
\end{equation}

\noindent The result of Eqn.~\ref{eq:T_5a} exhibits some unique characteristics. It is observed that the qudit pair is entangled with respect to state 
$\Ket{31_4}$ since the measurement of $\Ket{\theta_4}$ resulting in $\Ket{\theta_4}=\Ket{1_4}$ implies that $\Ket{\phi_4}=\Ket{3_4}$.  However, if 
$\Ket{\theta_4}$ is measured and results in $\Ket{\theta_4}=\Ket{0_4}$, there remains an equally likely chance that $\Ket{\phi_4}$ has assumed
a basis value of $\Ket{0_4}$, $\Ket{1_4}$, or $\Ket{2_4}$.  This is an example of a partially entangled pair. The output states for all possible input combinations of Fig.~\ref{fig:part-ent-gen}a are found in the center column of the leftmost partial entanglement results in Table~\ref{tb:par_full_ent_data}.

\begin{figure}[!t]
\includegraphics[height=0.75in]{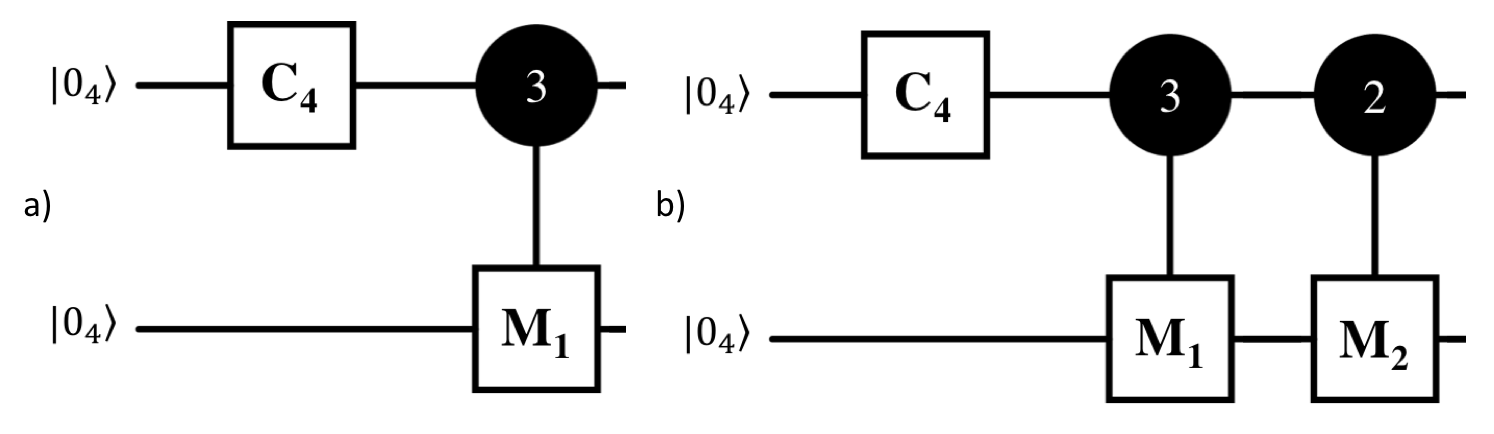}
\centering
\caption{Example radix-4 partial entanglement generators}
\label{fig:part-ent-gen}
\end{figure}

The radix-4 partial entanglement generator of Fig.~\ref{fig:part-ent-gen}a can be expanded to create Fig.~\ref{fig:part-ent-gen}b. This second circuit includes an additional $\mathbf{A}_{h,k}$ operator, and a different type of partially entangled radix-4 state results:

\begin{equation} \label{eq:T_5b}
\scriptsize{
\begin{split}
\mathbf{T}_{Fig.~\ref{fig:part-ent-gen}b}\Ket{00_4}&=\mathbf{A}_{2,2}\times\mathbf{A}_{3,1}\times (\mathbf{C}_4 \otimes \mathbf{I}_4)\Ket{00_4} \\
&=\mathbf{A}_{(2,3),(2,1)}\times (\mathbf{C}_4 \otimes \mathbf{I}_4)\Ket{00_4} \\
&=\frac{1}{2}\left[\Ket{00_4} + \Ket{10_4} + \Ket{22_4} + \Ket{31_4} \right] \\
&= \frac{1}{2}\left[ \left(\Ket{0_4} + \Ket{1_4}\right) \otimes\Ket{0_4}\right]  + \frac{1}{2}\left[\Ket{22_4}+ \Ket{31_4} \right]. 
\end{split}}
\end{equation}

\noindent In Eqn.~\ref{eq:T_5b}, the generated state shows an increase in the number of unique, fully-entangled basis state values and a decrease in the number of terms that have a factorable form and are indicative of the presence of partial entanglement.  The output states for all possible input combinations of Fig.~\ref{fig:part-ent-gen}b are found in the third column of leftmost set of partial entanglement results in Table~\ref{tb:par_full_ent_data}.

When the radix-4 QIP system contains fewer than $r-1 = 4-1 = 3$ $\mathbf{A}_{h,k}$ gates, partially entangled states are generated. The degree of partial entanglement, or the number of basis states with a common factor, is directly related to the total number of $\mathbf{A}_{h,k}$ operations in the partial entanglement circuit. When including multiple $\mathbf{A}_{h,k}$ gates in a partial-entanglement circuit, each value of $h$ and $k$ must differ.

Although partially entangled states may not become as commonly implemented as fully entangled states in higher-radix systems, that does not mean they are simply an interesting phenomena without a use case. Possible implementations could involve generating either entangled states or random values in a controlled manner for higher-radix QIP communication systems. For instance, observing one of the entangled basis states results in known second qudit's value, but observing an unentangled qudit that is factorable leaves an even distribution with respect to the possible observed state of the second qudit. In addition, since there are many combinations of partial entanglement generators, the output quantum state can be customized for a particular application.

\begin{center}
\begin{table*}[!htbp]
\renewcommand\arraystretch{1.2}
\caption{Output of Radix-4 Partial and Full Entanglement Generator Circuits in Figs.~\ref{fig:part-ent-gen} and \ref{fig:full-ent-gen} }
\begin{tabular}{lr}
{\scriptsize
\hfill{}
\begin{tabular}{|c|c|c|}
\hline
& \multicolumn{2}{|c|}{\textbf{Two-Qudit Operation(s) in Partial Entanglement Generator }}\\
\cline{2-3}
\hline
 \textbf{Input}&$\mathbf{A}_{3,1}$ & $\mathbf{A}_{2,2}\times\mathbf{A}_{3,1}=\mathbf{A}_{(2,3),(2,1)}$  \\
 \hline
$\Ket{00_4}$ &$ \frac{1}{2}\left[\Ket{00_4} + \Ket{10_4} + \Ket{20_4} + \Ket{31_4} \right]$ &$ \frac{1}{2}\left[\Ket{00_4} + \Ket{10_4} + \Ket{22_4} + \Ket{31_4} \right]$\\
$\Ket{01_4}$ &$ \frac{1}{2}\left[\Ket{01_4} + \Ket{11_4} + \Ket{21_4} + \Ket{32_4} \right]$& $ \frac{1}{2}\left[\Ket{01_4} + \Ket{11_4} + \Ket{23_4} + \Ket{32_4} \right]$ \\
$\Ket{02_4}$ &$ \frac{1}{2}\left[\Ket{02_4} + \Ket{12_4} + \Ket{22_4} + \Ket{33_4} \right]$& $ \frac{1}{2}\left[\Ket{02_4} + \Ket{12_4} + \Ket{20_4} + \Ket{33_4} \right]$ \\
$\Ket{03_4}$ &$ \frac{1}{2}\left[\Ket{03_4} + \Ket{13_4} + \Ket{23_4} + \Ket{30_4} \right]$ &$ \frac{1}{2}\left[\Ket{03_4} + \Ket{13_4} + \Ket{21_4} + \Ket{30_4} \right]$\\
$\Ket{10_4}$ & $ \frac{1}{2}\left[\Ket{00_4} + i\Ket{10_4} - \Ket{20_4} -i \Ket{31_4} \right]$ & $ \frac{1}{2}\left[\Ket{00_4} + i\Ket{10_4} -\Ket{22_4} - i\Ket{31_4} \right]$\\
$\Ket{11_4}$ &$ \frac{1}{2}\left[\Ket{01_4} + i\Ket{11_4} - \Ket{21_4} - i\Ket{32_4} \right]$& $ \frac{1}{2}\left[\Ket{01_4} + i\Ket{11_4} - \Ket{23_4} -i \Ket{32_4} \right]$ \\
$\Ket{12_4}$ &$ \frac{1}{2}\left[\Ket{02_4} + i\Ket{12_4} - \Ket{22_4} -i \Ket{33_4} \right]$& $ \frac{1}{2}\left[\Ket{02_4} + i\Ket{12_4} - \Ket{20_4} -i \Ket{33_4} \right]$ \\
$\Ket{13_4}$ &$ \frac{1}{2}\left[\Ket{03_4} + i\Ket{13_4}- \Ket{23_4} -i \Ket{30_4} \right]$ &$ \frac{1}{2}\left[\Ket{03_4} + i\Ket{13_4} - \Ket{21_4} -i \Ket{30_4} \right]$\\
$\Ket{20_4}$ & $ \frac{1}{2}\left[\Ket{00_4} - \Ket{10_4} + \Ket{20_4} - \Ket{31_4} \right]$ &$ \frac{1}{2}\left[\Ket{00_4} - \Ket{10_4} + \Ket{22_4} - \Ket{31_4} \right]$ \\
$\Ket{21_4}$ &$ \frac{1}{2}\left[\Ket{01_4} - \Ket{11_4} + \Ket{21_4} - \Ket{32_4} \right]$&  $ \frac{1}{2}\left[\Ket{01_4} - \Ket{11_4} + \Ket{23_4} - \Ket{32_4} \right]$\\
$\Ket{22_4}$ &$ \frac{1}{2}\left[\Ket{02_4} - \Ket{12_4} + \Ket{22_4} - \Ket{33_4} \right]$&$ \frac{1}{2}\left[\Ket{02_4} - \Ket{12_4} + \Ket{20_4} - \Ket{33_4} \right]$   \\
$\Ket{23_4}$ &$ \frac{1}{2}\left[\Ket{03_4} - \Ket{13_4} + \Ket{23_4} - \Ket{30_4} \right]$  &$ \frac{1}{2}\left[\Ket{03_4} - \Ket{13_4} + \Ket{21_4} - \Ket{30_4} \right]$\\
$\Ket{30_4}$ & $ \frac{1}{2}\left[\Ket{00_4} -i \Ket{10_4} - \Ket{20_4} +i \Ket{31_4} \right]$ & $ \frac{1}{2}\left[\Ket{00_4} -i \Ket{10_4} - \Ket{22_4}+i  \Ket{31_4} \right]$\\
$\Ket{31_4}$ &$ \frac{1}{2}\left[\Ket{01_4} -i\Ket{11_4} - \Ket{21_4} +i \Ket{32_4} \right]$&$ \frac{1}{2}\left[\Ket{01_4} -i \Ket{11_4} - \Ket{23_4} +i \Ket{32_4} \right]$   \\
$\Ket{32_4}$ &$ \frac{1}{2}\left[\Ket{02_4} -i \Ket{12_4} - \Ket{22_4} +i \Ket{33_4} \right]$& $ \frac{1}{2}\left[\Ket{02_4} -i\Ket{12_4} - \Ket{20_4} +i \Ket{33_4} \right]$ \\
$\Ket{33_4}$ & $ \frac{1}{2}\left[\Ket{03_4} -i \Ket{13_4} - \Ket{23_4} +i \Ket{30_4} \right]$&$ \frac{1}{2}\left[\Ket{03_4} -i \Ket{13_4} - \Ket{21_4} +i \Ket{30_4} \right]$\\
\hline
\end{tabular}}
\hfill{}
&
\renewcommand\arraystretch{1.0}
{\scriptsize
\hfill{}
\begin{tabular}{|c|c|}
\hline
\multicolumn{2}{|c|}{\textbf{Two-Qudit Operation(s) in Full Entanglement Generator }}\\
\cline{1-2}
\hline
 \textbf{Input} & $\mathbf{A}_{1,3}\times\mathbf{A}_{2,2}\times\mathbf{A}_{3,1}=\mathbf{A}_{(1,2,3),(3,2,1)}$  \\
 \hline
$\Ket{00_4}$ &$\frac{1}{2}\left[\Ket{00_4} + \Ket{13_4} + \Ket{22_4} + \Ket{31_4}\right]$ \\
$\Ket{01_4}$ & $\frac{1}{2}\left[\Ket{01_4} + \Ket{10_4} + \Ket{23_4} + \Ket{32_4}\right]$\\
$\Ket{02_4}$ &$\frac{1}{2}\left[\Ket{02_4} + \Ket{11_4} + \Ket{20_4} + \Ket{33_4}\right]$ \\
$\Ket{03_4}$ & $\frac{1}{2}\left[\Ket{03_4} + \Ket{12_4} + \Ket{21_4} + \Ket{30_4}\right]$\\
$\Ket{10_4}$ & $\frac{1}{2}\left[\Ket{00_4} + i\Ket{13_4} - \Ket{22_4} -i \Ket{31_4}\right]$ \\
$\Ket{11_4}$ &$\frac{1}{2}\left[\Ket{01_4} + i\Ket{10_4} -\Ket{23_4} -i \Ket{32_4}\right]$ \\
$\Ket{12_4}$ &$\frac{1}{2}\left[\Ket{02_4} +i \Ket{11_4} - \Ket{20_4} -i \Ket{33_4}\right]$ \\
$\Ket{13_4}$ &$\frac{1}{2}\left[\Ket{03_4} +i \Ket{12_4} - \Ket{21_4} -i \Ket{30_4}\right]$ \\
$\Ket{20_4}$ & $\frac{1}{2}\left[\Ket{00_4} -\Ket{13_4} + \Ket{22_4} - \Ket{31_4}\right]$   \\
$\Ket{21_4}$ &$\frac{1}{2}\left[\Ket{01_4} - \Ket{10_4} + \Ket{23_4} - \Ket{32_4}\right]$\\
$\Ket{22_4}$ &$\frac{1}{2}\left[\Ket{02_4} - \Ket{11_4} + \Ket{20_4} - \Ket{33_4}\right]$  \\
$\Ket{23_4}$ & $\frac{1}{2}\left[\Ket{03_4} - \Ket{12_4} + \Ket{21_4} - \Ket{30_4}\right]$\\
$\Ket{30_4}$ & $\frac{1}{2}\left[\Ket{00_4} -i\Ket{13_4} - \Ket{22_4} +i \Ket{31_4}\right]$  \\
$\Ket{31_4}$ & $\frac{1}{2}\left[\Ket{01_4} -i \Ket{10_4} - \Ket{23_4} +i \Ket{32_4}\right]$  \\
$\Ket{32_4}$ &$\frac{1}{2}\left[\Ket{02_4} -i \Ket{11_4} - \Ket{20_4} +i \Ket{33_4}\right]$\\
$\Ket{33_4}$ & $\frac{1}{2}\left[\Ket{03_4} -i \Ket{12_4} - \Ket{21_4} +i \Ket{30_4}\right]$\\
\hline
\end{tabular}}
\hfill{}
\label{tb:par_full_ent_data}
\end{tabular}
\end{table*}
\end{center}

\section{Demonstration of Full Entanglement}

The previous section described how partially entangled qudit pair generators can be created with a quantum circuit that includes a radix-$r$ Chrestenson gate and fewer than $r-1$ single-controlled-modulo-add-by-$k$ gates. In accordance with Theorem \ref{thm:max-ent-thm},  a total of appropriate $r-1$ controlled Modulo-add operations are required for full entanglement to occur.

\begin{theorem} \textit{Maximal Entanglement Generator} \label{thm:max-ent-thm}
A maximally entangled radix-$r$ qudit pair can be generated when the the pair
is initialized to a basis state and one of the qudits evolves to a state of maximal superposition through the application of a Chrestenson gate, $\mathbf{C}_r$. 
This superimposed qudit is then applied to the control inputs of $r-1$ controlled modulo-add-by-$k$ gates, $\mathbf{A}_{h,k}$,
and the other non-superimposed qudit is applied to the target.  
Each of the control values, $h$, of the $r-1$ $\mathbf{A}_{h,k}$ gates has a separate
and distinct value from the set $\{0,1,\cdots,(r-1)\}$ and each of the 
modulo-add-by-$k$ target operations, $k$, of the $r-1$ $\mathbf{A}_{h,k}$ gates, takes on
a  separate and distinct value from the set $\{1,\cdots,(r-1)\}$.
\end{theorem}

\begin{proof}
Consider Eqn. \ref{eq:T_5a} wherein a single controlled modulo-add-by-$k$ gate, $\mathbf{A}_{3,1}$ results
in a single fully entangled state and $r-1$ partially entangled states.  Likewise, in Eqn. \ref{eq:T_5b} another
controlled modulo-add-by-$k$ gate, $\mathbf{A}_{2,2}$ is applied resulting in $r-2$ fully entangled states and
$r-2$ partially entangled states.  Hence, by induction it is observed that the application of $r-1$
controlled modulo-add-by-$k$ gates,  where all control values $h$ are unique from the set $\{0,1,\cdots,(r-1)\}$
and all modulo-addition constants $k$ are unique from the set $\{1,\cdots,(r-1)\}$, results in $r$ fully-entangled qudit basis states. \hfill $\Box$
\end{proof}

In the case of the example radix-4 QIP system, three distinct $\mathbf{A}_{h,k}$ gates  are required to generate full entanglement where each set of $h$ and $k$ values must contain unique integers. An illustration of a radix-4 full entanglement generator is given in Fig.~\ref{fig:full-ent-gen}. In this circuit, the $\mathbf{C}_4$ gate creates maximal control qudit superposition and the two-qudit operation  $\mathbf{A}_{1,3}\times\mathbf{A}_{2,2}\times\mathbf{A}_{3,1}=\mathbf{A}_{(1,2,3),(3,2,1)}$ generates maximal entanglement among the qudit pair. When the initialized qudit state $\Ket{00_4}$ is evolved via the circuit in Fig.~\ref{fig:full-ent-gen},  a fully entangled qudit pair results.

\begin{equation} \label{eq:T_6}
\scriptsize{
\begin{split}
\mathbf{T}_{Fig.~\ref{fig:full-ent-gen}}\Ket{00_4}&= \mathbf{A}_{1,3}\times\mathbf{A}_{2,2}\times\mathbf{A}_{3,1} \times (\mathbf{C}_4 \otimes \mathbf{I}_4)\Ket{00_4} \\
&=\mathbf{A}_{(1,2,3),(3,2,1)}\times (\mathbf{C}_4 \otimes \mathbf{I}_4)\Ket{00_4} \\
&= \frac{1}{2}\left[\Ket{00_4} + \Ket{13_4} + \Ket{22_4} + \Ket{31_4} \right]. \\
\end{split}}
\end{equation}

In Eqn.~\ref{eq:T_6}, all of the output value basis states have unique values for both qudits. Thus, the qudits are fully entangled. When other initial basis states of qudit pairs evolve through the example entanglement generator of Fig. \ref{fig:full-ent-gen}, alternative fully entangled qudit states result. These different fully entangled states resulting from the Fig.~\ref{fig:full-ent-gen} circuit are given in rightmost set of fully entangled results in Table~\ref{tb:par_full_ent_data}.  We note that this set of states is analogous to the radix-2 Bell states when the radix is extended to $r=4$.  Therefore, one result of this paper is the generalization of the Bell state generator to a radix-$r$ system.

The qudit entanglement generator pictured in Fig.~\ref{fig:full-ent-gen} is not the only structure that creates maximally entangled radix-4 qudits. A $\mathbf{C}_4$ gate followed by any group of three different $\mathbf{A}_{h,k}$ gates, each with unique values for both $h$ and $k$, will create entangled states from qudits originally in a basis state. Since there are $r=4$ options for the control level, $h$, the three $\mathbf{A}_{h,k}$ gates needed for the full entanglement generator have a total of four different combinations for the implemented control values: $(\Ket{1_4}, \Ket{2_4}, \Ket{3_4})$, $(\Ket{0_4}, \Ket{1_4}, \Ket{2_4})$, $(\Ket{0_4}, \Ket{1_4}, \Ket{3_4})$, and $(\Ket{0_4}, \Ket{2_4}, \Ket{3_4})$. Considering the three $k$ values for the $r=4$ modulo-add-by-$k$ operations, there are six permutations for each control group. This gives a total of 24 radix-4 full entanglement generators. The order in which the $\mathbf{A}_{(h,k)}$ gates appear after the $\mathbf{C}_4$ gate is irrelevant, thus there are six orders for each of the 24 different full entanglement generators, resulting in 144 different and distinct $r=4$ full entanglement generators.  
In general, there are 

\begin{equation} \label{eq:total_max_ent_gen}
\prod_{i=2}^{r} (i^2-i)
\end{equation}

\noindent different full entanglement circuits of the form described here for a pair of radix-$r$ qudits. Since combining $\mathbf{A}_{(h,k)}$ gates is commutative, $\mathbf{A}_{(h_1,k_1)} \times \mathbf{A}_{(h_2,k_2)} =\mathbf{A}_{(h_2,k_2)} \times \mathbf{A}_{(h_2,k_2)} $, the function

\begin{equation} \label{eq:total_unique_max_ent_gen}
\prod_{i=2}^{r} \frac{(i^2-i)}{(i-1)} = \prod_{i=2}^{r} \frac{i(i-1)}{(i-1)} =  \prod_{i=2}^{r} i = r!
\end{equation}

\noindent is used to determine the number of maximal entanglement circuit configurations that produce unique transfer functions.

\begin{figure}[!t]
\includegraphics[height=0.75in]{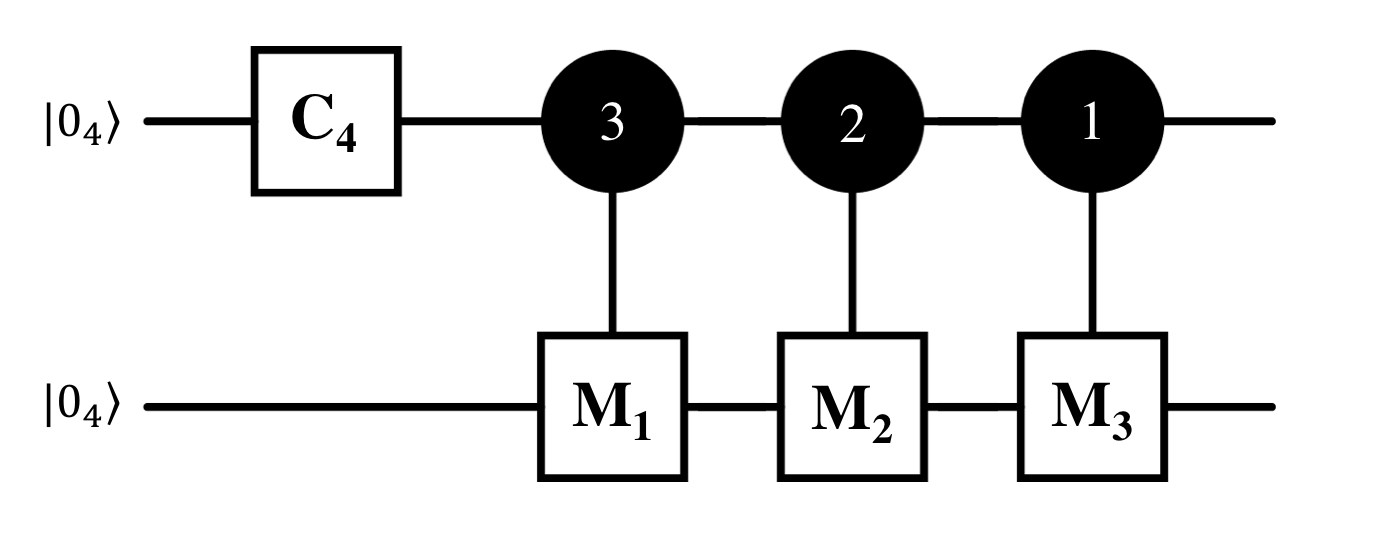}
\centering
\caption{Radix-4 full entanglement generator}
\label{fig:full-ent-gen}
\end{figure}

\section{Conclusion}
We have introduced and defined the concept of partial entanglement for higher-dimensional QIS systems.
In binary QIS, entanglement is either present or not. There is no concept of partial entanglement.
In a higher-dimensioned QIS system, entanglement can be present in varying degrees. 
We have furthermore developed
entanglement generators for higher-dimensional QIS systems.  These generators, when they are configured to produce maximal entanglement, are directly analogous to the well known Bell state generators in binary qubit-based QIS.  We have also shown how entanglement generators can be configured to produce partially entangled states. The partial and maximal entanglement state generators developed and described here use single-qudit Chrestenson operators for maximal superposition and controlled modulo-add-by-$k$ operators for entangling.  
These results will enable algorithms in binary QIS to be generalized and extended to higher-dimensional qudit-based systems.

\end{document}